\documentclass[11pt]{article}

\usepackage[T1]{fontenc}
\usepackage{microtype}
\usepackage{amsmath,amsthm}
\usepackage{newpxtext,newpxmath}
\usepackage[margin=15pt,font=small,labelfont={bf}]{caption}
\usepackage[margin=1in]{geometry}
\usepackage[english]{babel}
\usepackage[utf8]{inputenc}

\usepackage[backend=biber,style=alphabetic,maxalphanames=3,minalphanames=3,maxbibnames=99,minbibnames=99]{biblatex}
\usepackage{graphicx}
\usepackage{algorithm}
\usepackage[noend]{algpseudocode}
\usepackage[belowfloats,bottom]{footmisc}
\usepackage{csquotes}

\AtBeginDocument{
  }
\usepackage{xcolor}
\usepackage[colorlinks=true,
    citecolor=darkgreen,
    linkcolor=darkred
]{hyperref}
\usepackage{aliascnt}
\definecolor{darkgreen}{rgb}{0,0.5,0.5}
\definecolor{darkred}{rgb}{0.5,0,0}
\newcommand\drop[1]{}
\newcommand*{\newaliastheorem}[3]{
  \newaliascnt{#1}{#2}
  \newtheorem{#1}[#1]{#3}
  \aliascntresetthe{#1}
  \expandafter\newcommand\csname #1autorefname\endcsname{#3}
}
\newtheorem{theorem}{Theorem}
\newaliastheorem{proposition}{theorem}{Proposition}
\newaliastheorem{lemma}{theorem}{Lemma}
\newaliastheorem{corollary}{theorem}{Corollary}
\newaliastheorem{definition}{theorem}{Definition}
\newaliastheorem{conjecture}{theorem}{Conjecture}
\newaliastheorem{example}{theorem}{Example}
\newaliastheorem{remark}{theorem}{Remark}

\addto\extrasenglish{

}

\title{Online Graph Coloring with Random Arrivals}
\author{
Hirotaka Yoneda\thanks{Supported by JSPS JP25K24465, JST ASPIRE JPMJAP2302, and JST ACT-X JPMJAX25CT.} \\
The University of Tokyo, Japan \\
squar37@gmail.com
\and
Masataka Yoneda\thanks{Supported by JSPS JP25K24465 and JST ASPIRE JPMJAP2302.} \\
The University of Tokyo, Japan \\
e869120@gmail.com
}
\date{}

\begin{document}
\maketitle

\pagenumbering{gobble}

\begin{abstract}
We study online graph coloring in the random-arrival model, where the vertices are presented in a uniformly random order. Assuming $n$ and $k$ are known in advance, we present a randomized exponential-time algorithm that colors any $k$-colorable graph on $n$ vertices with $O(k^3 \log^2 k \log^4 n)$ colors in expectation. This is the first online coloring algorithm that uses a polylogarithmic number of colors for all fixed $k \geq 3$, and is better than a known lower bound of $\Omega(\log^{k-1} n)$ colors under adversarial arrivals for all fixed $k \geq 6$.

In addition, assuming $n$ is known in advance, we obtain a competitive ratio of $\widetilde{O}(n^{2/3})$. For comparison, it is known to be NP-hard to approximate (offline) graph coloring within a factor of $n^{1-\varepsilon}$ for all fixed $\varepsilon > 0$. Therefore, our exponential-time online algorithm achieves a performance guarantee beyond the NP-hardness threshold for polynomial-time approximation algorithms.
\end{abstract}

\newpage

\pagenumbering{arabic}
\setcounter{page}{1}

\section{Introduction}

We study the \emph{online coloring} problem, in which vertices are revealed one by one. When a vertex arrives, we must assign it a color immediately, before the next vertex arrives, so that its color differs from the colors of all its previously revealed neighbors. See \autoref{fig:101} for an example. This is one of the most extensively studied problems in online algorithms, alongside the $k$-server problem, paging/caching, and online matching \cite{Alb03}.

\begin{figure}[b]
    \centering
    \includegraphics[width=\linewidth]{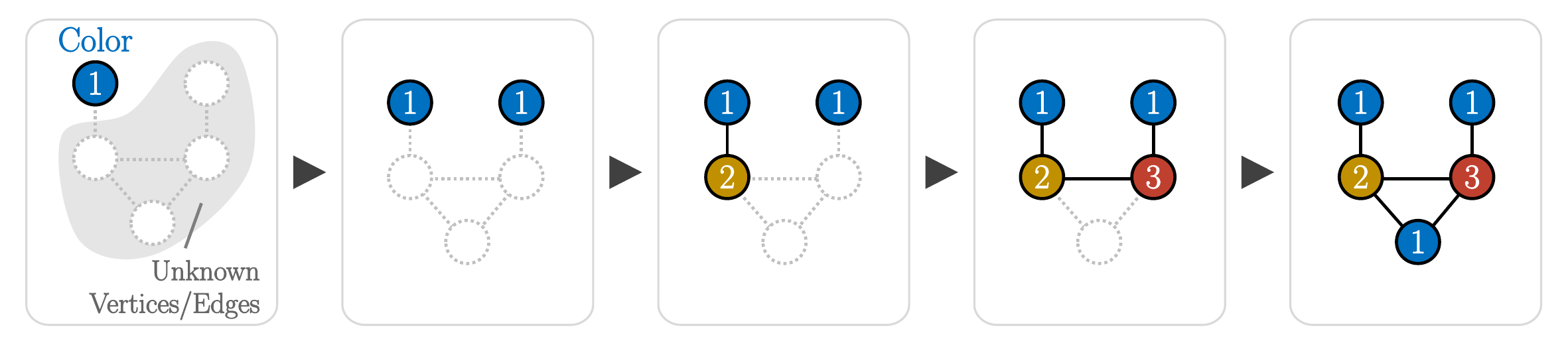}
    \caption{An example of online graph coloring. The numbers inside the vertices indicate their assigned colors. Vertices and edges drawn with gray dotted lines have not yet been revealed. The player must assign a color immediately when a vertex arrives.}
    \label{fig:101}
\end{figure}

\subsection{Online Coloring}

Online coloring is a classical problem for which natural algorithms are ineffective. The most intuitive strategy is \emph{First-Fit} \cite{GL88}, which always chooses the smallest available color. However, First-Fit can use $n/2$ colors even on an $n$-vertex bipartite graph, as shown in \autoref{fig:102} (such graphs are called crown graphs\footnote{A crown graph on $n$ vertices is a complete bipartite graph $K_{n/2,n/2}$ with a perfect matching removed. In particular, we consider the graph $G = (\{u_1, v_1, \dots, u_{n/2}, v_{n/2}\}, E)$ where $E = \{\{u_i, v_j\} : i \neq j\}$, and the vertices arrive in the order $u_1, v_1, \dots, u_{n/2}, v_{n/2}$.}). Accordingly, online coloring has been widely studied since the 1980s \cite{LST89}. The best-known upper bounds on the competitive ratio are as follows: $O(\frac{n}{\log n})$ for randomized algorithms by Halld\'{o}rsson (1997) \cite{Hal97}, and $O(\frac{n}{\log \log n})$ for deterministic algorithms by Kawarabayashi, Yoneda, and Yoneda (2025) \cite{KYY25}. In contrast, a lower bound of $\Omega(\frac{n}{\log^2 n})$ is known due to Halld\'{o}rsson and Szegedy (1994) \cite{HS94}.

\begin{figure}[t]
    \centering
    \includegraphics[width=\linewidth]{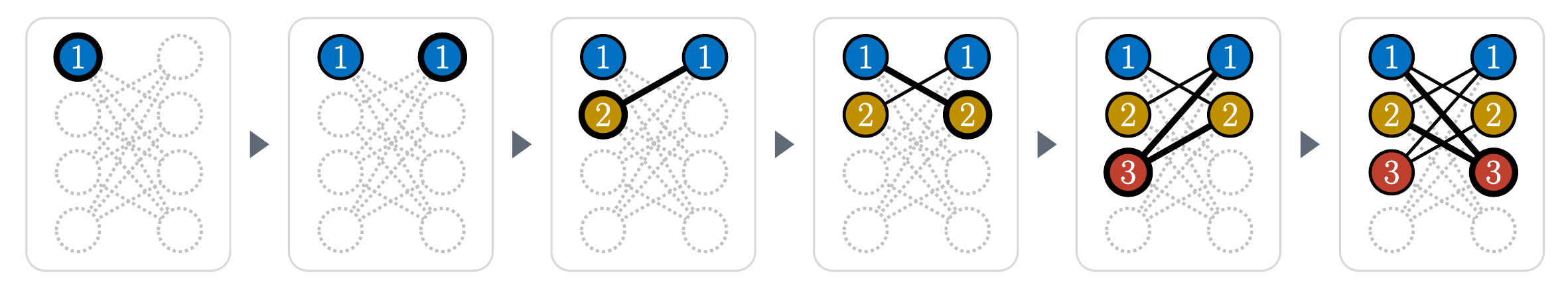}
    \caption{An $n$-vertex bipartite graph on which First-Fit uses $n/2$ colors.}
    \label{fig:102}
\end{figure}

Given these strong negative results, online coloring of $k$-colorable graphs for a fixed $k \geq 2$ has received attention. For bipartite graphs ($k = 2$), a deterministic algorithm by Lov\'{a}sz, Saks, and Trotter (1989) \cite{LST89} achieves $O(\log n)$ colors. The situation becomes much harder for $k \geq 3$: even for $k = 3$, the best-known results are a randomized $\widetilde{O}(n^{1/2})$-color algorithm by Vishwanathan (1992) \cite{Vis92} and a deterministic $\widetilde{O}(n^{2/3})$-color algorithm by Kierstead (1998) \cite{Kie98}. For comparison, the best-known lower bound is $\Omega(\log^{k-1} n)$ colors \cite{Vis92}, leaving a substantial gap between the upper and lower bounds.

Online coloring is also studied for special graph classes, e.g., trees \cite{Bea76}, interval graphs \cite{KT81}, planar graphs \cite{Ira94}, bounded-treewidth graphs \cite{AS21}, and graphs of large odd girth \cite{YY26}. See \cite{Kie05} for a survey on online coloring.

\subsection{The Random-Order Model}

The \emph{random-order model}, which assumes that the input order is chosen uniformly at random, is widely studied in the context of online algorithms. The most well-known problem is arguably the \emph{secretary problem}, in which $n$ candidates arrive for a job interview in random order, and the interviewer must decide immediately after each interview whether to hire the candidate. Only one candidate can be hired overall. Thanks to random arrivals, the classical secretary algorithm \cite{Lin61} selects the best candidate with probability at least $\frac{1}{e} \approx 0.368$. In contrast, under an adversarial arrival order, no algorithm can guarantee a success probability greater than $\frac{1}{n}$. The random-order model also leads to improved performance in many other online problems, including the following examples (see \cite{GS20} for a survey):

\begin{itemize}
    \item For the online bipartite matching problem, the optimal (tight) randomized competitive ratio is $1 - \frac{1}{e} \approx 0.632$ under adversarial arrivals \cite{KVV90}. However, under the random-order model, an algorithm achieves a competitive ratio of $0.696$ \cite{MY11}.
    
    \item For the online knapsack problem, no algorithm can achieve a bounded competitive ratio under adversarial arrivals. However, under the random-order model, an algorithm achieves a competitive ratio of $0.153$ \cite{KK26}.
\end{itemize}

\paragraph{Online Coloring.} In online coloring with random arrivals, the given graph $G = (V, E)$ is fixed, and the arrival order of the $n$ vertices is chosen uniformly at random from the $n!$ possibilities, each with probability $\frac{1}{n!}$. This setting has been particularly well-studied for First-Fit.

\begin{itemize}
    \item For trees, First-Fit uses $\Theta(\log n)$ colors under adversarial arrivals \cite{Bea76, GL88}, but it drops to $\Theta(\frac{\log n}{\log \log n})$ colors in expectation under the random-order model \cite{FGK+24, BGL+24}.
    
    \item For crown graphs (\autoref{fig:102}), First-Fit uses $n/2$ colors under adversarial arrivals, but it drops to $2.5 + O(\frac{1}{n})$ colors in expectation under the random-order model \cite{YXW+26}.
    
    \item For general bipartite graphs, it was conjectured that First-Fit only uses polylogarithmic number of colors in the random-order model \cite{BGL+24}, but it is claimed to be disproved; a very recent preprint shows a lower bound of $\Omega(n^{1/2000})$ colors in expectation \cite{Li26}.
\end{itemize}

However, due to the inherent difficulty of online coloring, few studies have investigated random arrivals beyond First-Fit. In particular, it remained open whether algorithms for $k$-colorable graphs could be improved under the random-order model; for example, whether an $O(n^{1/2-\varepsilon})$-color algorithm exists for 3-colorable graphs, which would improve upon the result of \cite{Vis92}.

\subsection{Our Contribution}

In this paper, we show that the random-order model substantially improves the performance in online coloring: $k$-colorable graphs can be colored online with a \emph{polylogarithmic} number of colors for any constant $k$. Formally, we prove the following main theorem.

\begin{theorem} \label{thm:main}
    In the random-order model, when both $n$ and $k$ are known in advance, there exists a randomized online algorithm that colors any $k$-colorable $n$-vertex graph using $O(k^3 \log^2 k \log^4 n)$ colors in expectation.
\end{theorem}

This result significantly improves upon the previous $\widetilde{O}(n^{1-\frac{1}{k-1}})$-color bound \cite{Vis92}. In addition, for all fixed $k \geq 6$, our random-order result is better than the known lower bound of $\Omega(\log^{k-1} n)$ colors for adversarial arrivals \cite{Vis92}. Therefore, we establish the \emph{asymptotic separation} between the random-order model and the adversarial-order model.

Moreover, assuming the result of the recent preprint \cite{Li26}, which claims that First-Fit uses expected $\Omega(n^{1/2000})$ colors for some bipartite graphs under random arrivals (and hence, for every $k \geq 2$, some $k$-colorable graphs), our algorithm asymptotically outperforms First-Fit for all $k \geq 2$.

Furthermore, we obtain a competitive ratio of $\widetilde{O}(n^{2/3})$, as stated in the following theorem.

\begin{theorem} \label{thm:main2}
    In the random-order model, when $n$ is known in advance, there exists a randomized online algorithm that attains a competitive ratio of $\widetilde{O}(n^{2/3})$.
\end{theorem}

Given the known lower bound of $\Omega(\frac{n}{\log^2 n})$ on the competitive ratio for adversarial arrivals \cite{HS94}, this also establishes the \emph{asymptotic separation} between the random-order model and the adversarial-order model in terms of the competitive ratio. Note that our result assumes that $n$ is known in advance. This is a standard assumption in the random-order model; several previous studies \cite{Alb21a, Alb21b, Lin61, San25} also make the same assumption.

\subsection{Online Algorithms vs. Approximation Algorithms}

A broader question in theoretical computer science concerns the relationship between the \emph{competitive ratio} in online algorithms and the \emph{approximation ratio} in polynomial-time approximation algorithms. This topic is discussed in various studies \cite{KKL07, MAS14}. In general, if an online algorithm achieves a performance guarantee beyond the hardness barrier for polynomial-time approximation algorithms, then it must run in super-polynomial time.

As a matter of fact, online algorithms can achieve better performance guarantees than approximation algorithms for some problems, although we did not find any references that explicitly mention this observation. The work of Rubinstein (2016) \cite{Rub16} presents an $O(\log n)$-competitive algorithm for the (maximum cardinality) secretary problem on a downward-closed set system. This immediately yields an $O(\log n)$-competitive algorithm for ``online maximum independent set'' with random arrivals. In contrast, approximating (offline) maximum independent set within a factor of $n^{1-\varepsilon}$ is NP-hard for every fixed $\varepsilon > 0$, as shown by Zuckerman (2006) \cite{Zuc06}.

This paper gives an exponential-time $\widetilde{O}(n^{2/3})$-competitive algorithm for online coloring with random arrival order (\autoref{thm:main2}), whereas approximating (offline) graph coloring within a factor of $n^{1-\varepsilon}$ is NP-hard for every fixed $\varepsilon > 0$ \cite{Zuc06}. Therefore, we establish the \emph{asymptotic separation} between online and (deterministic or randomized) approximation algorithms in graph coloring, under the widely believed assumption that $\textsf{NP} \not\subseteq \textsf{BPP}$ \cite{AB09}. We note that, although a graph coloring can be obtained by repeatedly choosing a maximum independent set among the uncolored vertices, Rubinstein's algorithm \cite{Rub16} does not directly imply an online coloring algorithm for random arrivals because the random arrival order on the uncolored vertices is not guaranteed.

\begin{remark}
    Since the approximation side allows randomized algorithms, including random-order processing, we regard this comparison as fair. Moreover, our result achieves $O(\log^4 n)$ colors for 3-colorable graphs; however, no polynomial-time algorithm that colors 3-colorable graphs with $O(n^{1/10})$ colors is known at the time of writing.\footnote{Bansal, Huang, and Lee (2026) \cite{BHL26} present an algorithm that colors 3-colorable graphs with $O(n^{0.1954})$ colors in polynomial time. Very recently, Song (2026) \cite{Son26} published a preprint claiming to achieve $O(n^{0.1653})$ colors.}
\end{remark}

\subsection{Overview of the Algorithm} \label{subsec:overview}

In this subsection, we present an overview of our algorithm and the intuition behind using random arrival order. For simplicity, we consider the case $k = 3$ here; the goal is to color a 3-colorable graph $G = (V, E)$ online with $O(\log^4 n)$ colors.

First, it suffices to devise an algorithm that colors the last $n/2$ vertices with $O(\log^3 n)$ colors, using information from the first $n/2$ vertices. If this can be done, we can color the entire graph with $O(\log^4 n)$ colors by recursively applying the same algorithm to the first $n/2$ vertices (for example, we color the $n/4+1, \dots, n/2$-th vertices using information from the first $n/4$ vertices), using completely disjoint color sets for different levels of recursion. Here, we denote the set of the first $n/2$ vertices by $S$.

\paragraph{The Benefit of Random Order.} The benefit of random arrival order is that $S \subseteq V$ can be assumed to be chosen uniformly at random. Hence, we can assume that the two parts, $G[S]$ and $G[V \setminus S]$, have similar properties. For example, the sizes of maximum independent sets and the sizes of each color class in a 3-coloring $\varphi: V \to \{1, 2, 3\}$ are very similar.

In particular, for a sufficiently large independent set $I \subseteq V$, the parts of $I$ that lie in $S$ and in $V \setminus S$ do not differ much in size, with a statistical error of only $\pm \widetilde{O}(\sqrt{n})$. Therefore, after the first $n/2$ vertices have arrived, if we ``sample'' an independent set $I \subseteq S$ of $G[S]$, then it is likely that $I$ can be ``extended'' to the vertices in $V \setminus S$ to the extent as large as $I$. This mechanism allows us to obtain a large independent set in $G[V \setminus S]$ online.

\paragraph{Hardcore Sampling and Extension.} After receiving all vertices in $S$, we sample an independent set $I \subseteq S$ from the \emph{hardcore model} \cite{DK25, Wei06} with parameter $\lambda \approx \frac{1}{40 \ln n}$, where $I$ is chosen with probability proportional to $\lambda^{|I|}$. As implied by the fact that a 3-colorable graph $G[S]$ has an independent set of size at least $|S|/3 \approx n/6$, the expected size of $I$ is $\Theta(\lambda n) = \Theta(\frac{n}{\log n})$.

Then, the set of ``eligible'' vertices for extending the independent set $I$ toward $V \setminus S$ is $J := \{v \in V \setminus S : N_I(v) = \emptyset\}$, which is sufficiently large by the discussion above.\footnote{Although we do not directly prove that the ``extension'' is sufficiently large in the subsequent proof, we believe that this fact helps explain the intuition behind the algorithm.} While $J$ itself may not be an independent set, it can be proven that $G[J]$ most probably has maximum degree $O(\log^2 n)$. See \autoref{fig:104} for an illustration. This implies that a large fraction of vertices in $V \setminus S$ can be colored online using First-Fit with only $O(\log^2 n)$ colors.

\begin{figure}[t]
    \centering
    \includegraphics[width=\linewidth]{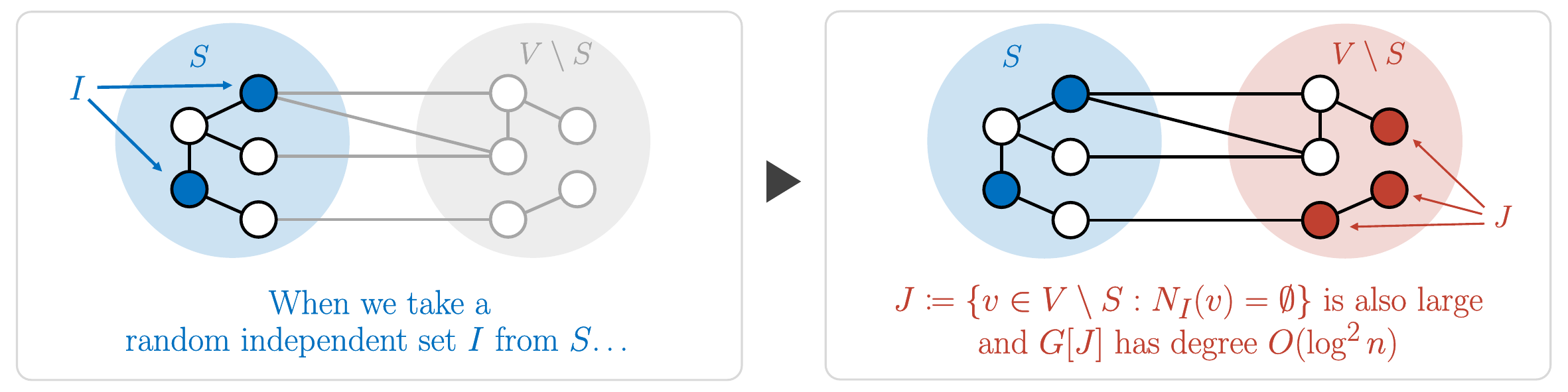}
    \caption{Sampling an independent set from $S$ and attempting to extend it into $V \setminus S$.}
    \label{fig:104}
\end{figure}

\paragraph{Covering $V \setminus S$ with Multiple Samples.} The issue is that, with a single independent set $I$, only a fraction of the vertices in $V \setminus S$ are colored. We address this issue by \emph{carefully} sampling a sequence of independent sets $I_1, \dots, I_r \ (I_r = \emptyset)$, where $r$ is most probably $O(\log n)$. We also maintain ``bins'' $J_1, \dots, J_r \subseteq V \setminus S$ for the extension counterpart, initially empty.

For each arrival of $v \in V \setminus S$, we find the smallest $i$ such that $N_{I_i}(v) = \emptyset$, and put $v$ into $J_i$. Then, we can prove that each $G[J_i]$ most probably has maximum degree $O(\log^2 n)$. Hence, by applying First-Fit to each of $G[J_1], \dots, G[J_r]$, using a separate set of colors for each bin, we can most probably color the entire $V \setminus S$ with $r \cdot O(\log^2 n) = O(\log^3 n)$ colors.

\subsection{Organization}

\autoref{sec:prelim} presents some formal definitions and notation. \autoref{sec:hardcore} introduces the hardcore model and investigates its properties. \autoref{sec:algorithm} introduces the main algorithm, and \autoref{sec:analysis} analyzes the number of colors used. \autoref{sec:competitive-ratio} presents an algorithm with a competitive ratio of $\widetilde{O}(n^{2/3})$. Finally, \autoref{sec:conclusion} concludes the paper and presents some conjectures.


\section{Preliminaries} \label{sec:prelim}

Before moving on to our algorithms, we present the required notation and definitions. First, we formally define online coloring with random arrival order. Then, we define First-Fit, the most straightforward strategy for online coloring.

\begin{definition}[Online Coloring]
    In the online coloring problem, an undirected graph $G = (V, E)$ is initially empty; i.e., there are zero vertices. The vertices of $G$ are given one by one. Immediately after a vertex $v$ arrives (along with all edges incident to previously arrived vertices), we must assign a color $c(v) \in \mathbb{N}$ that differs from the colors of all adjacent vertices. The goal is to minimize the number of colors used, $|\{c(v) : v \in V\}|$.
\end{definition}

\begin{definition}[Random-Order Model] \label{def:random-order}
    The random-order model of online coloring assumes that a graph $G^* = (V^*, E^*)$, where $V^* = \{v_1, \dots, v_n\}$, is initially fixed and is hidden from the algorithm. We additionally assume that $n$ \textbf{is known in advance}. The vertices arrive in the order $v_{\sigma(1)}, \dots, v_{\sigma(n)}$, where $\sigma: \{1, \dots, n\} \to \{1, \dots, n\}$ is a permutation that is selected uniformly at random from $n!$ possibilities, each with probability $\frac{1}{n!}$.
\end{definition}

\begin{definition}[First-Fit]
    First-Fit is an algorithm for online coloring that assigns each new vertex the smallest available color.
\end{definition}

We use the following notation. Given a graph $G = (V, E)$, we denote the set of vertices adjacent to $v \in V$ as $N(v)$, and the set of vertices adjacent to at least one vertex in $S \subseteq V$ as $N(S)$. For a vertex set $X \subseteq V$, we denote $N_X(v) := N(v) \cap X$ and $N_X(S) := N(S) \cap X$. We denote the subgraph induced by $S \subseteq V$ as $G[S]$. We denote the chromatic number of $G$ as $\chi(G)$. We denote the set of all independent sets as $\mathcal{I}(G)$. In $O$-, $\Omega$-, and $\Theta$-notation, we write $\widetilde{O}(f(n))$, $\widetilde{\Omega}(f(n))$, and $\widetilde{\Theta}(f(n))$ to hide polylogarithmic factors in $f(n)$.


\section{Hardcore Model} \label{sec:hardcore}

In this section, we discuss the properties of randomly sampled independent sets. We sample independent sets using the \emph{hardcore model} \cite{DK25, Wei06}.

\begin{definition}[Hardcore Model]
    Given a graph $G = (V, E)$, the hardcore model $\textsf{IS}_{\lambda}(G)$ with parameter $\lambda > 0$ is a probability distribution that chooses an independent set $I \in \mathcal{I}(G)$ with probability proportional to $\lambda^{|I|}$. More precisely, define
    \begin{equation*}
        Z_{\lambda}(G) := \sum_{I \in \mathcal{I}(G)} \lambda^{|I|}.
    \end{equation*}
    Then, an independent set $I \in \mathcal{I}(G)$ is chosen with probability $\frac{\lambda^{|I|}}{Z_{\lambda}(G)}$. See \autoref{fig:201} for an example.
\end{definition}

\begin{figure}[b]
    \centering
    \includegraphics[width=\linewidth]{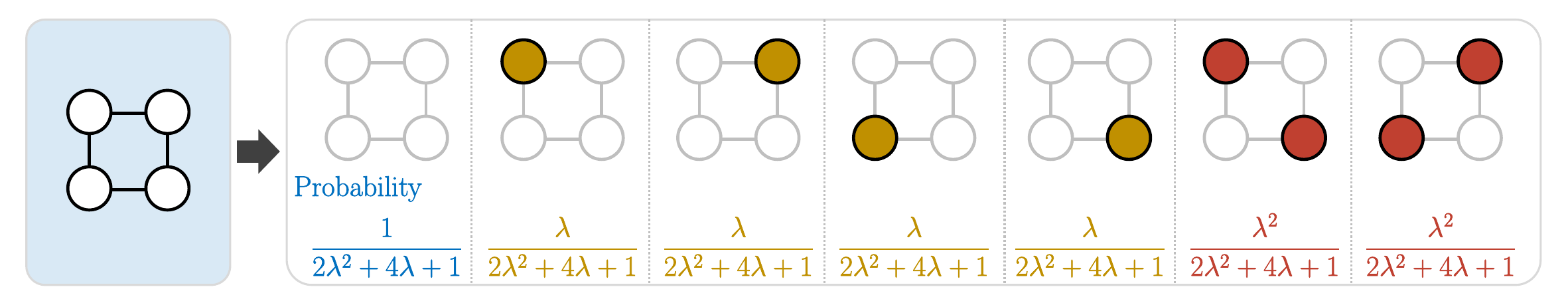}
    \caption{An example of the hardcore model, with the probability of selecting each independent set.}
    \label{fig:201}
\end{figure}

\subsection{Basic Properties}

First, we present some properties of the hardcore model that are used in the subsequent proofs.

\begin{proposition} \label{prop:hardcore-1}
    Let $\lambda > 0$, and let $G = (V, E)$ be a $k$-colorable graph. Then,
    \begin{equation*}
        (1+\lambda)^{|V|/k} \leq Z_{\lambda}(G) \leq (1+\lambda)^{|V|}.
    \end{equation*}
\end{proposition}

\vspace{1.0em}

\begin{proof}
    For the upper bound, $Z_{\lambda}(G) \leq \sum_{S \subseteq V} \lambda^{|S|} = (1+\lambda)^{|V|}$ by the binomial theorem. For the lower bound, since $G$ has a $k$-coloring, one color class has size at least $|V|/k$, so it contains an independent set $I' \subseteq V$ with $|I'| \geq |V|/k$. By the binomial theorem,
    \begin{equation*}
        Z_{\lambda}(G)
        \geq \sum_{I \subseteq I'} \lambda^{|I|}
        = (1+\lambda)^{|I'|}
        \geq (1+\lambda)^{|V|/k}.
        \qedhere
    \end{equation*}
\end{proof}

\begin{proposition} \label{prop:hardcore-2}
    Let $\lambda > 0$, and let $G = (V, E)$ be a graph. For disjoint vertex sets $S, T \subseteq V \ (S \cap T = \emptyset)$,
    \begin{equation*}
        Z_{\lambda}(G[S \cup T]) \leq Z_{\lambda}(G[S]) \cdot (1+\lambda)^{|T|}.
    \end{equation*}
\end{proposition}

\begin{proof}
    Each $I \in \mathcal{I}(G[S \cup T])$ can be written as $I = I_S \cup I_T$, where $I_S \in \mathcal{I}(G[S])$ and $I_T \in \mathcal{I}(G[T])$. Therefore, $Z_{\lambda}(G[S \cup T]) \leq Z_{\lambda}(G[S]) \cdot Z_{\lambda}(G[T])$. By \autoref{prop:hardcore-1}, we have $Z_{\lambda}(G[T]) \leq (1+\lambda)^{|T|}$, which proves the desired inequality.
\end{proof}

\subsection{The Size of Independent Sets}

We analyze the expected size of an independent set. Note that the case $\lambda = 1$ was already studied by Nilli (1988) \cite{Nil88}, who showed that the average size of independent sets for a $k$-colorable graph is $\Omega(\frac{n}{k \log k})$ and proved that this bound is tight.

\begin{lemma} \label{lem:hardcore-3}
    Let $\lambda \in (0, 1]$, and let $G = (V, E)$ be a $k$-colorable graph. Then,
    \begin{equation*}
        \mathbb{E}_{I \sim \mathsf{IS}_{\lambda}(G)}[|I|] \geq \frac{\lambda}{3k \ln (3k)} |V|.
    \end{equation*}
\end{lemma}

\begin{proof}
    First, the function $f(x) = (\frac{1}{3k})^x$ is convex, so it follows from Jensen's inequality that:
    \begin{align*}
        \left(\frac{1}{3k}\right)^{\mathbb{E}[|I|]}
        & \leq \mathbb{E}\left[\left(\frac{1}{3k}\right)^{|I|}\right] \\
        & = \frac{1}{Z_{\lambda}(G)} \cdot \left(\sum_{I \in \mathcal{I}(G)} \left(\frac{1}{3k}\right)^{|I|} \lambda^{|I|}\right)
        = \frac{Z_{\lambda/3k}(G)}{Z_{\lambda}(G)}
        \leq \frac{(1+\lambda/3k)^{|V|}}{(1+\lambda)^{|V|/k}}
    \end{align*}
    The last inequality is due to \autoref{prop:hardcore-1}. Therefore:
    \begin{align*}
        \mathbb{E}[|I|] & \geq -\log_{3k} \frac{(1+\lambda/3k)^{|V|}}{(1+\lambda)^{|V|/k}} \\
        & = \frac{1}{\ln (3k)} \left\{\frac{|V|}{k} \ln(1+\lambda) - |V| \ln \left(1 + \frac{\lambda}{3k}\right)\right\} \\
        & \geq \frac{1}{\ln (3k)} \left\{\frac{|V|}{k} \cdot \frac{2}{3} \lambda - |V| \cdot \frac{\lambda}{3k}\right\}
        = \frac{\lambda}{3k \ln (3k)} |V|
    \end{align*}
    The last inequality follows from $\frac{2}{3} x \leq \ln (1+x) \leq x$ for all $x \in [0, 1]$.
\end{proof}

\subsection{The Neighborhood Structures}

We also investigate the structure surrounding a random independent set $I \sim \textsf{IS}_{\lambda}(G)$. In the following two lemmas, we analyze the expected size of $I \cup N(I)$.

\begin{lemma} \label{lem:hardcore-4}
    Let $\lambda > 0$, and let $G = (V, E)$ be a graph. Then,
    \begin{equation*}
        \mathbb{E}_{I \sim \textsf{IS}_{\lambda}(G)}[|I|] = \lambda \cdot  \mathbb{E}_{I \sim \textsf{IS}_{\lambda}(G)}[|V \setminus (I \cup N(I))|].
    \end{equation*}
\end{lemma}

\begin{figure}[!b]
    \centering
    \includegraphics[width=\linewidth]{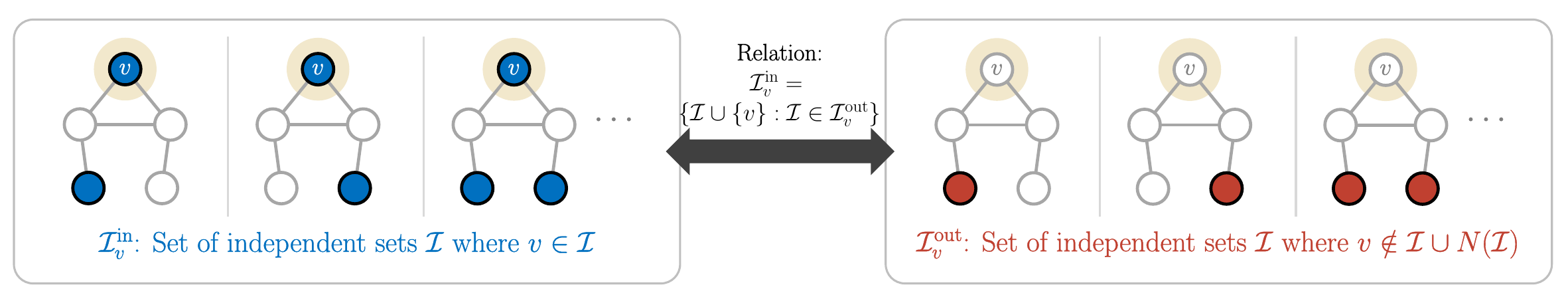}
    \caption{The definitions of $\mathcal{I}_v^{in}$ and $\mathcal{I}_v^{out}$, and the bijection between them.}
    \label{fig:302}
\end{figure}

\begin{proof}
    Fix $v \in V$. Let $\mathcal{I}^{\mathrm{in}}_v = \{I \in \mathcal{I} : v \in I\}$, and $\mathcal{I}^{\mathrm{out}}_v = \{I \in \mathcal{I} : v \in V \setminus (I \cup N(I))\}$. Then, the relationship $\mathcal{I}^{\mathrm{in}}_v = \{I \cup \{v\} : I \in \mathcal{I}^{\mathrm{out}}_v \}$ holds (\autoref{fig:302}), as shown below:

    \begin{itemize}
        \item For every $I \in \mathcal{I}^{\mathrm{out}}_v$, $I \cup \{v\} \in \mathcal{I}^{\mathrm{in}}_v$ because $I \cup \{v\}$ is an independent set due to $v \notin I \cup N(I)$.
        \item For every $I \cup \{v\}\in \mathcal{I}^{\mathrm{in}}_v$, $I \in \mathcal{I}^{\mathrm{out}}_v$ because $I$ is an independent set and $v \notin N(I)$.
    \end{itemize}

    \noindent
    Here, for each $I \in \mathcal{I}^{\mathrm{out}}_v$, the probability that $I \cup \{v\}$ is chosen is $\lambda$ times the probability that $I$ is chosen. Therefore,
    \begin{align*}
        & \mathrm{Pr}_{I \sim \mathsf{IS}_{\lambda}(G)}[I \in \mathcal{I}^{\mathrm{in}}_v] = \lambda \cdot \mathrm{Pr}_{I \sim \mathsf{IS}_{\lambda}(G)}[I \in \mathcal{I}^{\mathrm{out}}_v]  \quad (\forall v \in V) \\
        \text{Equivalently,} \ & \mathrm{Pr}_{I \sim \mathsf{IS}_{\lambda}(G)}[v \in I] = \lambda \cdot \mathrm{Pr}_{I \sim \mathsf{IS}_{\lambda}(G)}[v \in V \setminus (I \cup N(I))] \quad (\forall v \in V)
    \end{align*}
    Summing this equation over all $v \in V$, we have $\mathbb{E}_{I \sim \textsf{IS}_{\lambda}(G)}[|I|] = \lambda \cdot  \mathbb{E}_{I \sim \textsf{IS}_{\lambda}(G)}[|V \setminus (I \cup N(I))|]$ by linearity of expectation.
\end{proof}

\begin{lemma} \label{lem:hardcore-5}
    Let $\lambda \in (0, 1]$, and let $G = (V, E)$ be a $k$-colorable graph. Then,
    \begin{equation*}
        \mathbb{E}_{I \sim \textsf{IS}_{\lambda}(G)}[|I \cup N(I)|] \leq \left(1 - \frac{1}{3k \ln(3k)}\right) |V|.
    \end{equation*}
\end{lemma}

\begin{proof}
    By \autoref{lem:hardcore-4}, $\mathbb{E}[|V \setminus (I \cup N(I))|] = \frac{1}{\lambda} \cdot \mathbb{E}[|I|]$. Thus, by \autoref{lem:hardcore-3},
    \begin{align*}
        \mathbb{E}[|I \cup N(I)|]
        & = |V| - \frac{1}{\lambda} \cdot \mathbb{E}[|I|] \\
        & \leq |V| - \frac{1}{\lambda} \cdot \frac{\lambda}{3k \ln (3k)} |V|
        = \left(1 - \frac{1}{3k \ln (3k)}\right) |V|. \qedhere
    \end{align*}
\end{proof}


\section{The Algorithm} \label{sec:algorithm}

We present an algorithm that colors $k$-colorable graphs online with $O(k^3 \log^2 k \log^4 n)$ colors in the random-order model. The intuition is already given in \autoref{subsec:overview}. Note that, in the subsequent sections, we use $G = (V, E)$ to denote the current graph for the algorithm and $G^* = (V^*, E^*)$ to denote the final graph, as in \autoref{def:random-order}.

\subsection{Basic Strategy}

We consider developing an algorithm that colors the last $n/2$ vertices using $O(k^3 \log^2 k \log^3 n)$ colors, given that the first $n/2$ vertices have already arrived and have been colored. If this can be done, we can recursively apply the same algorithm to the first $n/2$ vertices. Then, we can color the entire graph with $O(k^3 \log^2 k \log^4 n)$ colors. 

However, for convenience in the analysis, we do not exactly divide them into $n/2$ vertices each. Instead, we sample $m$ from a binomial distribution $\textsf{Binomial}(n, \frac{1}{2})$, and we color the last $n-m$ vertices, given that the first $m$ vertices have already arrived. We note that $m$ is very close to $n/2$ (with $\pm \widetilde{O}(\sqrt{n})$ error), so this does not fundamentally change the algorithm. We use disjoint colors for the first $m$ vertices and the last $n-m$ vertices. Let $S$ be the set of the first $m$ vertices. Under this framework, $S \subseteq V^*$ is uniformly distributed over the $2^n$ subsets of $V^*$, as shown below.

\begin{lemma} \label{prop:random-s}
    Let $m \sim \textsf{Binomial}(n, \frac{1}{2})$. Let $S$ be the set of the first $m$ vertices. Then, each subset of $V^*$ appears as $S$ with probability exactly $2^{-n}$.
\end{lemma}

\begin{proof}
    The probability that $|S| = m$ is $2^{-n} \cdot \binom{n}{m}$ by the binomial distribution. Also, there are $\binom{n}{m}$ possible subsets of $V^*$ with size $m$, and all of them have equal probabilities of appearing as $S$ because of the random-order model. Hence, \autoref{prop:random-s} holds.
\end{proof}

\subsection{Algorithm Details}

Now, we consider an algorithm to color the last $n-m$ vertices $V^* \setminus S$. The algorithm consists of two phases: (1) the \emph{learning phase}, which learns and extracts the structure of the first $m$ vertices, and (2) the \emph{output phase}, which actually colors the last $n-m$ vertices by exploiting the learned structure. The two phases are described as follows.

\paragraph{(1) The Learning Phase.} The core idea is that, if we sample an independent set $I \subseteq S$, many vertices in $V^* \setminus S$ are likely to have no neighbor in $I$ thanks to the random arrivals (\autoref{subsec:overview}). Such a set of vertices can be colored efficiently, as we prove in \autoref{sec:analysis}. However, a single sampled independent set $I$ does not cover all vertices in $V^* \setminus S$. Therefore, we need to sample multiple independent sets.

In the algorithm, we sample a sequence of independent sets $I_1, I_2, \dots$, each from a hierarchy $(S =) \ A_1 \supsetneq A_2 \supsetneq \cdots$. We start with $A_1 = S$. For $i = 1, 2, \dots$ in this order, we sample an independent set $I_i$ from $\textsf{IS}_{\lambda}(G[A_i])$, and then we let $A_{i+1} := N_{A_i}(I_i)$. See \autoref{fig:301} for an example. Eventually, $I_i$ becomes empty; we let $r$ be the first index such that $I_r = \emptyset$. When coloring the remaining $n-m$ vertices, we use independent sets $I_1, \dots, I_r$. It can be proven that $r = O(k \log k \log n)$ with high probability (\autoref{lem:rounds}), regardless of the value of the parameter $\lambda \in (0, 1]$.

\begin{figure}[t]
    \centering
    \includegraphics[width=\linewidth]{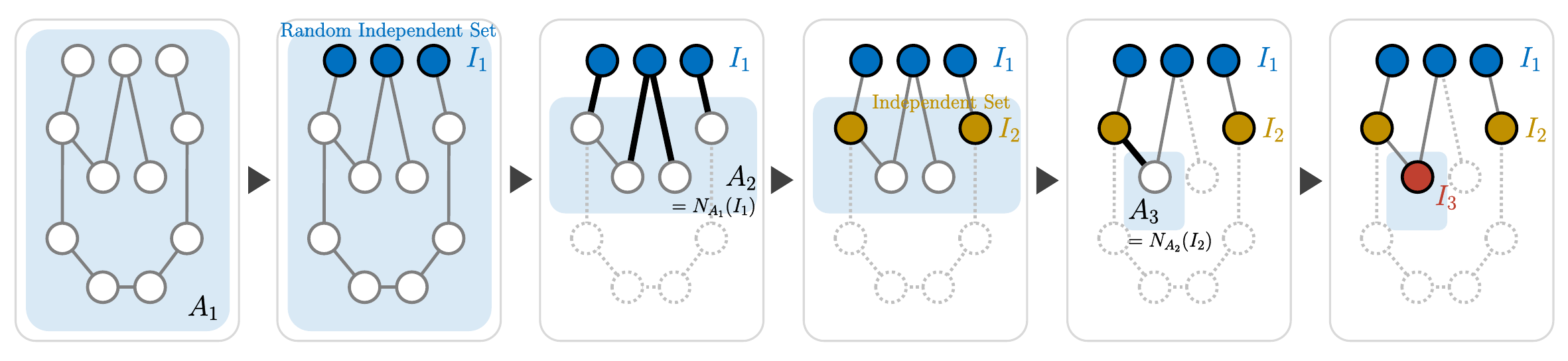}
    \caption{An example of the learning phase.}
    \label{fig:301}
\end{figure}

\paragraph{(2) The Output Phase.} Next, we color the remaining $n-m$ vertices. Here, we maintain bins $J_1, \dots, J_r \subseteq V^* \setminus S$, initially empty, which correspond to $I_1, \dots, I_r$, respectively.

When a vertex $v \in V^* \setminus S$ arrives, we put $v$ into $J_i$ for the smallest $i$ such that $N_{I_i}(v) = \emptyset$. Such an $i$ must exist because $I_r = \emptyset$. Although $J_i$ is not necessarily an independent set, it can be proven that the maximum degree of $G^*[J_i]$ is $O((k \log n) / \lambda)$ with high probability in the range of $i \leq \frac{1}{\lambda}$ (\autoref{lem:degree-prob-all}). Hence, with the parameter setting
\begin{equation*}
    \lambda := \frac{1}{6k \ln (3k) \ln n + 2},
\end{equation*}
$G^*[J_1], \dots, G^*[J_r]$ all have maximum degrees $O(k^2 \log k \log^2 n)$ with high probability (\autoref{lem:max-degree}). Therefore, we can use the First-Fit algorithm to color each vertex of $G^*[J_i]$. We use disjoint sets of colors for bins $J_1, \dots, J_r$.

\paragraph{Overall Algorithm.} The pseudocode for the entire algorithm is presented in \autoref{alg:main}. With high probability, we use $O(k^2 \log k \log^2 n)$ colors for each $J_i$, and $r = O(k \log k \log n)$. Therefore, we use $O(k^3 \log^2 k \log^3 n)$ colors in expectation for the last $n-m$ vertices (\autoref{lem:colors-part}). We analyze the algorithm's performance in detail in the next section.

\begin{algorithm}[t]
    \caption{$\textsf{COLORING}_k(n)$: An algorithm to color an $n$-vertex graph online.}
    \label{alg:main}
    \begin{algorithmic}[1]
        \State $\lambda \gets \frac{1}{6k \ln (3k) \ln n + 2}$
        \State $m \sim \textsf{Binomial}(n, \frac{1}{2})$
        \State Color the first $m$ vertices with $\textsf{COLORING}_k(m)$, if $m \geq 1$
        \State $S \gets$ the set of the first $m$ vertices
        \State $A_1 \gets S; r \gets 1$
        \Loop
            \State $I_r \gets \textsf{IS}_{\lambda}(G[A_r])$
            \If {$I_r = \emptyset$}
                \State \textbf{break}
            \EndIf
            \State $A_{r+1} \gets N_{A_r}(I_r)$
            \State $r \gets r + 1$
        \EndLoop
        \State $J_1, \dots, J_r \gets \emptyset$
        \For {each arrival of the remaining vertices $v \in V^* \setminus S$}
            \State $i \gets$ the smallest $i$ such that $N_{I_i}(v) = \emptyset$
            \State $J_i \gets J_i \cup \{v\}$, and color $v$ using First-Fit on $G[J_i]$
        \EndFor
    \end{algorithmic}
\end{algorithm}

\begin{remark}
    The runtime of \autoref{alg:main} is unbounded for the scenario where $m = n$ is chosen every time. However, if $m = n$, the only thing that happens is a call to the identical procedure $\textsf{COLORING}_k(n)$. Therefore, the outcome of \autoref{alg:main} remains the same if we choose $m$ from the distribution that excludes $m = n$, i.e., each $m = 0, \dots, n-1$ is chosen with probability $\frac{1}{2^n-1} \binom{n}{m}$; this change makes its runtime bounded. Note that $r$ is also bounded; we will prove $r \leq n+1$ in \autoref{lem:rounds-strict}. Nevertheless, this algorithm runs in exponential time, because sampling an independent set from $\textsf{IS}_{\lambda}(G[A_r])$ takes exponential time.
\end{remark}


\section{Analysis} \label{sec:analysis}

In this section, we analyze the performance of \autoref{alg:main}.

\subsection{Number of Levels}

First, we analyze the number $r$ of sampled independent sets: we prove that this is surely at most $n+1$ and is $O(k \log k \log n)$ with high probability.

\begin{lemma} \label{lem:rounds-strict}
    $r \leq n+1$.
\end{lemma}

\begin{proof}
    Since \autoref{alg:main} ensures that $I_1, \dots, I_{r-1} \neq \emptyset$ and $A_{i+1} = N_{A_i}(I_i)$ contains no vertex of $I_i$ for each $i$, it follows that $A_1 \supsetneq \dots \supsetneq A_r$. Since $|A_1| \leq n$, we have $r \leq n+1$.
\end{proof}

\begin{lemma} \label{lem:rounds}
    Suppose that $\lambda \in (0, 1]$. For all possible inputs of the first $m$ vertices,
    \begin{equation*}
        \mathrm{Pr}[r \leq 6k \ln (3k) \ln n + 2] \geq 1 - \frac{1}{n}.
    \end{equation*}
\end{lemma}

\begin{proof}
    Fix an arbitrary realization of $A_i$. By applying \autoref{lem:hardcore-5} to $G[A_i]$, we have
    \begin{align*}
        \mathbb{E}_{I \sim \textsf{IS}_{\lambda}(G[A_i])} [|A_{i+1}|]
        & = \mathbb{E}_{I \sim \textsf{IS}_{\lambda}(G[A_i])} [|N_{A_i}(I)|] \\
        & \leq \mathbb{E}_{I \sim \textsf{IS}_{\lambda}(G[A_i])} [|I \cup N_{A_i}(I)|] \\
        & \leq \left(1 - \frac{1}{3k \ln (3k)}\right) |A_i|.
    \end{align*}
    Taking expectations over $A_i$, we obtain $\mathbb{E}[|A_{i+1}|] \leq \left(1 - \frac{1}{3k \ln (3k)}\right) \mathbb{E}[|A_i|]$. Since $|A_1| \leq n$,
    \begin{equation*}
        \mathbb{E}[|A_{i+1}|] \leq \left(1 - \frac{1}{3k \ln (3k)}\right)^i \cdot n.
    \end{equation*}
    Since $(1-\frac{1}{x})^x < \frac{1}{e}$ for all $x \geq 1$, we have $\mathbb{E}[|A_{i^*+1}|] \leq \frac{1}{n}$ with $i^* := \lceil 3k \ln(3k) \cdot 2 \ln n \rceil$. Since $|A_{i^*+1}|$ is a non-negative integer, Markov's inequality implies that $A_{i^*+1}$ is empty with probability at least $1 - \frac{1}{n}$. Moreover, if $A_{i^*+1}$ is empty, then $I_{i^*+1}$ is also empty, meaning that $r \leq i^*+1$. Therefore, $r \leq 6k \ln (3k) \ln n + 2$ with probability at least $1 - \frac{1}{n}$.
\end{proof}

\subsection{Maximum Degrees}

Next, we analyze the maximum degrees of the subgraphs $G^*[J_i]$. Specifically, we show that the maximum degrees of $G^*[J_1], \dots, G^*[J_r]$ are $O(k^2 \log k \log^2 n)$ with high probability.

\subsubsection{Introducing the History Sequence}

Here, for a given graph $G^* = (V^*, E^*)$, we need to analyze all possibilities for the first $m$ vertices and the sampled independent sets. The outcome of \autoref{alg:main} in terms of coloring the last $n-m$ vertices, i.e., $J_1, \dots, J_r$, is determined by $(A_1, I_1, \dots, A_r, I_r)$. To this end, we define \emph{history sequences} as follows.

\begin{definition}[History Sequence]
    Let $t \geq 1$ be an integer. Given a graph $G^* = (V^*, E^*)$, the ``history sequence up to level $t$'' is defined as a sequence $\mathcal{H} := (A_1, I_1, \dots, A_t, I_t)$. This is defined even when $r < t$; in this case, $A_i = I_i = \emptyset$ when $i > r$. The set of all possible history sequences up to level $t$ is denoted as
    \[
    \begin{aligned}
        \mathbb{H}_t := \{(A_1,I_1,\ldots,A_t,I_t):\;&
        A_1 \subseteq V^*,\
        I_i \in \mathcal{I}(G^*[A_i])\ (1 \le i \le t),\\
        & A_{i+1}=N_{A_i}(I_i)\ (1 \le i < t)\}.
    \end{aligned}
    \]
    It is easy to see that $\mathbb{H}_t$ indeed enumerates all history sequences with positive probability.
\end{definition}

\begin{proposition} \label{prop:history}
    Let $t \geq 1$ be an integer. The probability that the history sequence up to level $t$ equals $\mathcal{H} = (A_1, I_1, \dots, A_t, I_t) \in \mathbb{H}_t$ is given by
    \begin{equation*}
        p(\mathcal{H}) := 2^{-n} \cdot \prod_{i=1}^t \frac{\lambda^{|I_i|}}{Z_{\lambda}(G^*[A_i])}.
    \end{equation*}
\end{proposition}

\begin{proof}
    The probability that $S = A_1$ is chosen is $2^{-n}$ by \autoref{prop:random-s}. The probability that $I_i$ is chosen from $\mathsf{IS}_{\lambda}(G^*[A_i])$ is $\frac{\lambda^{|I_i|}}{Z_{\lambda}(G^*[A_i])}$ by the definition of the hardcore model (note that this probability is $1$ for $i > r$). Thus, \autoref{prop:history} holds.
\end{proof}

\subsubsection{Degree Analysis with Double-Counting Method}

To establish the maximum degree bound, we prove the following lemma using the double-counting method on history sequences.

\begin{lemma} \label{lem:magic}
    Let $\lambda > 0$ be the parameter of the hardcore model. Suppose $G^* = (V^*, E^*)$ is a fixed $k$-colorable graph. For every vertex $v \in V^*$ and every integer $t \geq 1$,
    \begin{equation*}
        \mathbb{E}\left[\text{$\left\{1 + \frac{\lambda}{(1+\lambda)^t}\right\}^{|N(v) \cap J_t|/k}$ if $v \in J_t$, otherwise $0$}\right] \leq 1.
    \end{equation*}
\end{lemma}

\begin{proof}
    Let $\mathbb{H}' \subseteq \mathbb{H}_t$ be the set of all history sequences up to level $t$ in which the vertex $v$ is put into $J_t$. For each $\mathcal{H} = (A_1, I_1, \dots, A_t, I_t) \in \mathbb{H}'$, define the following:
    
    \begin{itemize}
        \item Let $X(\mathcal{H}) \subseteq V^*$ be the set of vertices adjacent to $v$ in $G^*[J_t]$.
        \item Let $Y(\mathcal{H}) \subseteq V^*$ be an arbitrary maximum independent set of $G^*[X(\mathcal{H})]$.
    \end{itemize}

    Since $G^*[X(\mathcal{H})]$ is $k$-colorable, we have $|Y(\mathcal{H})| \geq |X(\mathcal{H})| / k$. Moreover, each subset $W \subseteq Y(\mathcal{H})$ is also an independent set (see \autoref{fig:502}). Therefore, we can write down $2^{|Y(\mathcal{H})|}$ history sequences as follows (we add $W$ to each of $A_1, \dots, A_t$ and to $I_t$):
    \begin{equation*}
        \text{$F(\mathcal{H}, W) := (A_1 \cup W, I_1, \dots, A_{t-1} \cup W, I_{t-1}, A_t \cup W, I_t \cup W)$, for each $W \subseteq Y(\mathcal{H})$}.
    \end{equation*}
    The sequence $F(\mathcal{H}, W)$ is indeed a valid history sequence up to level $t$. This can be shown as follows. First, for every $i < t$, $I_i$ is an independent set of $G^*[A_i \cup W]$, because $I_i$ is an independent set of $G^*[A_i]$. Second, $A_{i+1} \cup W = N_{A_i \cup W}(I_i)$ holds, because if some vertex $w \in W$ has no neighbor in $I_i$, then $w$ should have been put into bin $J_i$ instead of $J_t$. Third, $I_t \cup W$ is an independent set because $W \subseteq Y(\mathcal{H})$ is an independent set and every vertex in $W \subseteq J_t$ has no neighbor in $I_t$.
    
\begin{figure}[t]
    \centering
    \includegraphics[width=\linewidth]{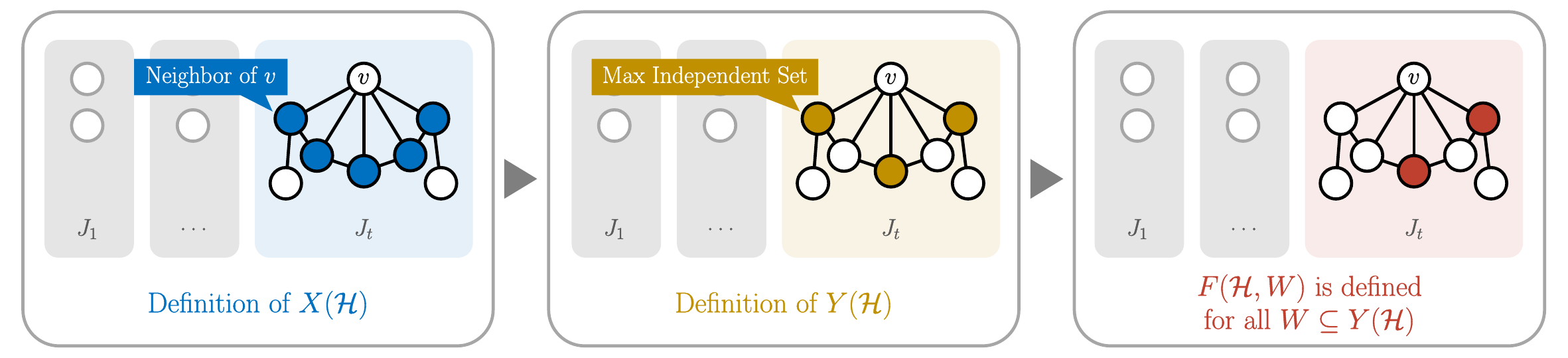}
    \caption{An illustration of the sets $X(\mathcal{H})$, $Y(\mathcal{H})$, and $W \subseteq Y(\mathcal{H})$ used to construct a new history sequence $F(\mathcal{H}, W)$.}
    \label{fig:502}
\end{figure}

    Next, the constructed histories $F(\mathcal{H}, W)$, over all $\mathcal{H} \in \mathbb{H}'$ and $W \subseteq Y(\mathcal{H})$, are pairwise distinct. This is because, from a history sequence $F(\mathcal{H}, W) = (A'_1, I'_1, \dots, A'_t, I'_t)$, we can uniquely recover $W$ as $I'_t \cap N(v)$, and hence we can also uniquely recover the original history sequence $\mathcal{H}$. The reason for $W = I'_t \cap N(v)$ is that, for $I'_t = I_t \cup W$, each vertex in $I_t$ is not adjacent to $v$ (due to $v \in J_t$), but each vertex in $W$ is adjacent to $v$ (due to $W \subseteq X(\mathcal{H})$). Therefore, the total probability of all history sequences $F(\mathcal{H}, W)$ is at most $1$:
    \begin{equation} \label{eq:prob-sum-1}
        \sum_{\mathcal{H} \in \mathbb{H}'} \sum_{W \subseteq Y(\mathcal{H})} p(F(\mathcal{H}, W)) \leq 1.
    \end{equation}

    On the other hand, we also have a lower bound on the probability $p(F(\mathcal{H}, W))$ in terms of $p(\mathcal{H})$. This is calculated as follows:
    \begin{align*}
        p(F(\mathcal{H}, W)) & = p(\mathcal{H}) \cdot \frac{\lambda^{|W|}}{\prod_{i=1}^t \{Z_{\lambda}(G^*[A_i \cup W]) / Z_{\lambda}(G^*[A_i])\}} \\
        & \geq p(\mathcal{H}) \cdot \frac{\lambda^{|W|}}{\{(1+\lambda)^{|W|}\}^t} \\
        & = p(\mathcal{H}) \cdot \left\{\frac{\lambda}{(1+\lambda)^t}\right\}^{|W|}.
    \end{align*}
    The first equality is derived from \autoref{prop:history}, and the inequality is derived from \autoref{prop:hardcore-2}. Therefore, by the binomial theorem, we obtain
    \begin{equation} \label{eq:prob-sum-2}
        \sum_{W \subseteq Y(\mathcal{H})} p(F(\mathcal{H}, W)) \geq p(\mathcal{H}) \cdot \left\{1 + \frac{\lambda}{(1+\lambda)^t}\right\}^{|Y(\mathcal{H})|}.
    \end{equation}
    Combining \autoref{eq:prob-sum-1}, \autoref{eq:prob-sum-2}, and $|Y(\mathcal{H})| \geq |X(\mathcal{H})| / k$, we obtain
    \begin{equation*}
        \sum_{\mathcal{H} \in \mathbb{H}'} \left\{p(\mathcal{H}) \cdot \left\{1 + \frac{\lambda}{(1+\lambda)^t}\right\}^{|X(\mathcal{H})| / k}\right\} \leq 1.
    \end{equation*}
    \autoref{lem:magic} follows directly from this inequality.
\end{proof}

\subsubsection{Deriving the Degree Bound}

\begin{lemma} \label{lem:degree-prob}
    Let $\lambda \in (0, 1]$ be the parameter of the hardcore model. Suppose $G^* = (V^*, E^*)$ is a fixed $k$-colorable graph. For every vertex $v \in V^*$ and every positive integer $t \leq \frac{1}{\lambda}$,
    \begin{equation} \label{eq:degree-prob-0}
        \mathrm{Pr}\left[\text{$v \in J_t$ and $|N(v) \cap J_t| \geq \frac{10k \ln n}{\lambda}$}\right] \leq \frac{1}{n^3}.
    \end{equation}
\end{lemma}

\begin{proof}
    First, for every positive integer $t \leq \frac{1}{\lambda}$, we have
    \begin{equation*}
        \frac{\lambda}{(1+\lambda)^t} \geq \frac{\lambda}{(1+\lambda)^{1/\lambda}} \geq \frac{\lambda}{e}.
    \end{equation*}
    Therefore, \autoref{lem:magic} implies that the following inequality also holds for $t \leq \frac{1}{\lambda}$:
    \begin{equation} \label{eq:degree-prob-1}
        \mathbb{E}\left[\text{$(1 + \lambda/e)^{|N(v) \cap J_t|/k}$ if $v \in J_t$, otherwise $0$}\right] \leq 1.
    \end{equation}
    By Markov's inequality, \autoref{eq:degree-prob-1} implies $\Pr[v \in J_t \ \text{and} \ (1 + \lambda/e)^{|N(v) \cap J_t|/k} \geq n^3] \leq \frac{1}{n^3}$. By taking the logarithm of both sides, the condition above becomes
    \begin{equation} \label{eq:degree-prob-2}
        |N(v) \cap J_t| \geq \frac{3k \ln n}{\ln(1 + \lambda/e)}.
    \end{equation}
    Here,
    \begin{equation} \label{eq:degree-prob-3}
        \frac{3k \ln n}{\ln(1 + \lambda/e)} \leq \frac{3k \ln n}{(3/10) \cdot \lambda} = \frac{10k \ln n}{\lambda}.
    \end{equation}
    This is due to $\ln(1+x/e) \geq \frac{3}{10} x$ for all $x \in [0, 1]$. Therefore, combining \autoref{eq:degree-prob-2} and \autoref{eq:degree-prob-3}, the probability that $v \in J_t$ and $|N(v) \cap J_t| \geq \frac{10k \ln n}{\lambda}$ is at most $\frac{1}{n^3}$.
\end{proof}

\begin{lemma} \label{lem:degree-prob-all}
    Let $\lambda \in (0, 1]$ be the parameter of the hardcore model, and let $t \leq \frac{1}{\lambda}$ be a positive integer. Suppose $G^* = (V^*, E^*)$ is a fixed $k$-colorable graph. The probability that the maximum degree of $G^*[J_t]$ exceeds $\frac{10 k \ln n}{\lambda}$ is at most $\frac{1}{n^2}$.
\end{lemma}

\begin{proof}
    The maximum degree of $G^*[J_t]$ exceeds $\frac{10 k \ln n}{\lambda}$ only if the event of \autoref{eq:degree-prob-0} occurs for at least one $v \in V^*$. Hence, by \autoref{lem:degree-prob} and the union bound, this probability is at most $\frac{1}{n^3} \cdot |V^*| = \frac{1}{n^2}$.
\end{proof}

\begin{lemma} \label{lem:max-degree}
    With $\lambda := \frac{1}{6k \ln(3k) \ln n + 2}$, for every fixed $k$-colorable graph $G^* = (V^*, E^*)$, $G^*[J_1], \dots, G^*[J_r]$ all have maximum degrees at most $10k \ln n \cdot (6k \ln (3k) \ln n + 2)$ with probability at least $1 - \frac{3}{n}$.
\end{lemma}

\begin{proof}
    We apply \autoref{lem:degree-prob-all} with parameter $\lambda := \frac{1}{6k \ln(3k) \ln n + 2}$. That is, for each $t \leq 6k \ln(3k) \ln n + 2$, the maximum degree of $G^*[J_t]$ exceeds $10k \ln n \cdot (6k \ln (3k) \ln n + 2)$ with probability at most $\frac{1}{n^2}$. Hence, the probability that both of the following conditions hold:

    \begin{itemize}
        \item $r \leq 6k \ln (3k) \ln n + 2$, and
        \item at least one of $G^*[J_1], \dots, G^*[J_r]$ has maximum degree more than $10k \ln n \cdot (6k \ln (3k) \ln n + 2)$
    \end{itemize}

    \noindent is at most $\frac{n+1}{n^2} \leq \frac{2}{n}$ by the union bound and the fact that $r \leq n+1$ (\autoref{lem:rounds-strict}). Excluding the case $r > 6k \ln (3k) \ln n + 2$, which happens with probability at most $\frac{1}{n}$ (\autoref{lem:rounds}), \autoref{lem:max-degree} follows.
\end{proof}

\subsection{The Grand Finale}

Finally, we analyze the total number of colors used in \autoref{alg:main}.

\begin{lemma} \label{lem:colors-part}
    In \autoref{alg:main}, the number of colors used for $V^* \setminus S$ is $O(k^3 \log^2 k \log^3 n)$ in expectation.
\end{lemma}

\begin{proof}
    If the maximum degrees in $G[J_1], \dots, G[J_r]$ are at most $\Delta$, then First-Fit uses at most $\Delta + 1$ colors, resulting in at most $r(\Delta + 1)$ colors for the entire $V^* \setminus S$. Therefore, by \autoref{lem:rounds}, \autoref{lem:max-degree}, and the union bound, we use at most $O(k \log k \log n) \cdot O(k^2 \log k \log^2 n) = O(k^3 \log^2 k \log^3 n)$ colors with probability at least $1 - \frac{4}{n}$. In the remaining case, which occurs with probability at most $\frac{4}{n}$, we may use up to $n$ colors, but this is negligible.
\end{proof}

\begin{lemma} \label{lem:recursion}
    The number of calls to $\textsf{COLORING}_k(\cdot)$ in \autoref{alg:main} is at most $2 \log_2 n + 1$ with probability at least $1 - \frac{1}{n}$.
\end{lemma}

\begin{proof}
    By \autoref{prop:random-s}, each vertex $v \in V^*$ is put into $S$, i.e., colored in the recursion, with probability $\frac{1}{2}$. Hence, each vertex $v \in V^*$ goes down for at least $2 \log_2 n$ recursions with probability at most $2^{-2 \log_2 n} = \frac{1}{n^2}$. If this does not happen for any $v \in V^*$, then the recursion depth in \autoref{alg:main} is at most $2 \log_2 n$. By the union bound, this occurs with probability at least $1 - \frac{1}{n^2} \cdot |V^*| = 1 - \frac{1}{n}$.
\end{proof}

We are now ready to prove our main theorem, which shows that \autoref{alg:main} uses at most $O(k^3 \log^2 k \log^4 n)$ colors in expectation.

\begingroup
\def\thetheorem{\ref{thm:main}}
\begin{theorem}[restated]
    In the random-order model, when both $n$ and $k$ are known in advance, there exists a randomized online algorithm that colors any $k$-colorable $n$-vertex graph using $O(k^3 \log^2 k \log^4 n)$ colors in expectation.
\end{theorem}
\addtocounter{theorem}{-1}
\endgroup

\begin{proof}
    For each call to $\textsf{COLORING}_k(\cdot)$, we use $O(k^3 \log^2 k \log^3 n)$ colors in expectation, due to \autoref{lem:colors-part}. Note that we can assume random arrivals in each call to $\textsf{COLORING}_k(m)$, because $m$ does not depend on the graph structure (i.e., it depends only on random choices), and the probability of this event does not change by shuffling the order of the first $m$ vertices, due to the random-order assumption in the original graph $G^*$.
    
    Hence, across the shallowest $\lfloor 2 \log_2 n + 1 \rfloor$ calls to $\textsf{COLORING}_k(\cdot)$ (or all calls if there are fewer), we use $O(k^3 \log^2 k \log^4 n)$ colors in expectation. If there are more than $2 \log_2 n + 1$ calls, at most $n$ colors may be used, but this probability is at most $\frac{1}{n}$ (\autoref{lem:recursion}), so this is negligible.
\end{proof}


\section{The Competitive Ratio} \label{sec:competitive-ratio}

We consider online coloring with random arrivals for general graphs. In online coloring, a lower bound of $\Omega(\frac{n}{\log^2 n})$ on the competitive ratio is known for both deterministic and randomized algorithms \cite{HS94}. However, with the random-order model, we can do better.

\begingroup
\def\thetheorem{\ref{thm:main2}}
\begin{theorem}[restated]
    In the random-order model, when $n$ is known in advance, there exists a randomized online algorithm that attains a competitive ratio of $\widetilde{O}(n^{2/3})$.
\end{theorem}
\addtocounter{theorem}{-1}
\endgroup

\begin{proof}
    The idea is to simulate the behavior of \autoref{alg:main} for $2^q$-colorable graphs, $\textsf{COLORING}_{2^q}(n)$, on the same input graph $G$, for each $q = 0, 1, \dots, \lceil \log_2 n \rceil$ in parallel processes. We assume disjoint sets of colors for different processes. For each arrival of vertex $v \in V$, we color $v$ using the color selected by the process that has used the fewest colors so far.

    We now analyze the number of colors used. Let $q^* := \lceil \log_2 \chi(G^*) \rceil$. Since $G^*$ is $2^{q^*}$-colorable, by \autoref{thm:main}, the number of colors used by the process with $q = q^*$, denoted by $c_{q^*}$, satisfies
    \begin{equation*}
        \mathbb{E}[c_{q^*}] = O((2^{q^*})^3 \log^2 (2^{q^*}) \log^4 n) = O(\chi(G^*)^3 \log^2 \chi(G^*) \log^4 n).
    \end{equation*}
    Moreover, for each process, at most $c_{q^*}$ colors are used in the \emph{actual} online coloring, so the number of colors used overall is at most $c_{q^*} (\lceil \log_2 n \rceil + 1)$. Therefore, this algorithm uses $O(\chi(G^*)^3 \log^2 \chi(G^*) \log^5 n)$ colors in expectation. Also, it is obvious that we use at most $n$ colors. Thus, we attain a competitive ratio of $O(n^{2/3} \log^{7/3} n)$, with the worst-case bound attained when $\chi(G^*) = \Theta(\frac{n^{1/3}}{\log^{7/3} n})$.
\end{proof}


\section{Conclusion} \label{sec:conclusion}

In this paper, we presented an algorithm for online coloring with random arrivals that uses $O(k^3 \log^2 k \log^4 n)$ colors for $k$-colorable graphs. For online coloring as a whole, we improved the performance for the case $k = 3$ from $\widetilde{O}(n^{1/2})$ colors \cite{Vis92} to $O(\log^4 n)$ colors. We also achieved a competitive ratio of $\widetilde{O}(n^{2/3})$, which is better than the conventional lower bound $\Omega(\frac{n}{\log^2 n})$ \cite{HS94} and is better than an $\Omega(n^{1-\varepsilon})$ inapproximability bound for graph coloring \cite{Zuc06}. These results demonstrate that the random-order model unlocks new possibilities for improving the performance of online coloring algorithms. It is now natural to pose the following conjectures.

\begin{conjecture}
    In the random-order model, when both $n$ and $k$ are known in advance, there exists a randomized online algorithm that colors any $k$-colorable $n$-vertex graph using $O(k \log n)$ colors in expectation.
\end{conjecture}

\begin{conjecture}
    In the random-order model, when $n$ is known in advance, there exists a randomized online algorithm that attains a competitive ratio of $O(\log n)$.
\end{conjecture}

Finally, a limitation of our result is that it relies on the random-order model and does not extend to adversarial arrivals. It is natural to ask whether the following conjecture holds.

\begin{conjecture}
    For any constant $k \geq 3$, there exists a (randomized) online algorithm for coloring $k$-colorable graphs with $n^{o(1)}$ colors, even without assuming random arrivals.
\end{conjecture}

\section*{Acknowledgements}

We are grateful to our professor Ken-ichi Kawarabayashi for helpful discussions and comments. We also thank Yuta Inoue for his feedback on the Introduction section.

\section*{AI Usage Disclosure}

The authors substantially used OpenAI’s large language model, GPT-6 Astra, during the proof exploration phase. Through an iterative process with the authors, the model generated the proof of the main theorem (\autoref{thm:main}) for the special case $k = 3$. Then, the authors simplified the proof, discussed and obtained the intuition behind the algorithm, and generalized the analysis to all $k \geq 3$. The authors also used GPT-6 Astra and GPT-5.6 Sol to polish the paper's presentation. The authors take full responsibility for the correctness and content of the paper.

\printbibliography

\end{document}